\documentclass[12pt]{article}

\usepackage{amsmath}
\usepackage{amssymb}
\usepackage{amsfonts}
\usepackage{latexsym}
\usepackage{color}

\catcode `\@=11 \@addtoreset{equation}{section}

\catcode `\@=12

\newcommand{\be}{\begin{equation}}
\newcommand{\en}{\end{equation}}
\newcommand{\ee}{\end{equation}}
\newcommand{\bea}{\begin{eqnarray}}
\newcommand{\ena}{\end{eqnarray}}
\newcommand{\beano}{\begin{eqnarray*}}
\newcommand{\enano}{\end{eqnarray*}}
\newcommand{\bee}{\begin{enumerate}}
\newcommand{\ene}{\end{enumerate}}

\newcommand{\N}{\mathfrak N}
\newcommand{\M}{\mathfrak M}

\newcommand{\mc}{\mathcal}

\newcommand{\E}{{\cal E}}
\newcommand{\F}{{\cal F}}

\newcommand{\1}{1 \!\! 1}

\newcommand{\Hil}{\mc H}

\newtheorem{thm}{Theorem}

\newtheorem{defn}[thm]{Definition}

\newenvironment{proof}{\noindent {\bf Proof}}{\hfill$\square$\vspace{3mm}\endtrivlist}

\catcode `\@=11 \@addtoreset{equation}{section}
\catcode `\@=12

\newcommand{\kt}{\rangle}
\newcommand{\br}{\langle}

\def\ben{$$}
\def\een{$$}
\def\ba{\begin{array}{c}}
\def\ea{\end{array}}

\begin{document}

\thispagestyle{empty}

\vspace*{1cm}

\begin{center}
{\Large \bf Non linear pseudo-bosons versus hidden Hermiticity. II:
The case of unbounded operators }
\\[10mm]

{\large Fabio Bagarello}\\
  Dipartimento di Metodi e Modelli Matematici,
Facolt\`a di Ingegneria,\\ Universit\`a di Palermo, I-90128  Palermo, Italy\\
e-mail: fabio.bagarello@unipa.it\\
home page:
www.unipa.it/fabio.bagarello\\

\vspace{3mm}

and

 \vspace{3mm}

{\large Miloslav Znojil}\\
  Nuclear Physics Institute ASCR,
 250 68 \v{R}e\v{z}, Czech Republic\\
e-mail: znojil@ujf.cas.cz\\
home page:
http://gemma.ujf.cas.cz/$\sim$znojil\\

\vspace{3mm}

\end{center}

\vspace*{2cm}

\begin{abstract}
 \noindent
A close parallelism between the notions of nonlinear pseudobosons
and of an apparent non-Hermiticity of observables as shown in paper
I is demonstrated  to survive the transition to the quantum models
using unbounded metric in the so called standard Hilbert space of
states.

\end{abstract}

\vspace{2cm}


\vfill


\newpage

\section{Introduction}

The core of the {\em difference} between the current bosonic and
fermionic quantum states is reflected by the respective number
operators with eigenvalues which may be any non-negative integer for
bosons and just zero or one for fermions. The most natural {\em
unification} of these states is being achieved under the notion of
supersymmetry \cite{Bagchi}. The latter concept finds a further
generalization in the models exhibiting the so called nonlinear
supersymmetry (NLSUSY, meaning, in essence \cite{Plyuschchay}, that
the anticommutator of the so called charges becomes equal to a
nonlinear polynomial function of the Hamiltonian) or, alternatively,
in the models composed of the so called nonlinear pseudobosons
(NLPB, \cite{bagpb1}).

There exists \cite{z3DI} a close relationship between the abstract
NLSUSY algebras and their representations in terms of certain
manifestly non-Hermitian operators (or, more explicitly
\cite{SIGMA}, cryptohermitian operators) of quantum observables with
real spectra. Remarkably enough, the latter observables may very
traditionally be selected as ordinary differential linear
Hamiltonians. In different context, their large subclass (called,
conveniently, ${\cal PT}-$symmetric Hamiltonians and sampled by the
Bessis' and Zinn-Justin's \cite{DB} $H = -\partial_x^2+{\rm i}x^3$)
has recently been made extremely popular by Carl Bender with
coauthors \cite{BB,Carl}).

In our preceding paper I \cite{bagzno} we demonstrated that there
also exists a similarly close connection between the {\em same}
class of the cryptohermitian Hamiltonian (of Hamiltonian-like)
operators $H\neq H^\dagger$ and the class of the generalized, NLPB
number operators $M\neq M^\dagger$. At the same time we felt it
rather unfortunate that the rigorous formulation of the expected
third possible connection between the NLPB systems and NLSUSY
algebras was still missing.

We saw one of the reasons in the emergence of a number of subtle
technical difficulties attributed to the unbounded-operator nature
of the Hamiltonians $H\neq H^\dagger$ which are needed in the NLSUSY
model building \cite{z1ZCBR}. As a consequence, our formulation of
the equivalence between the notions of the cryptohermiticity and
NLPB characteristics of quantum systems in paper I relied, heavily,
on the assumptions of the boundedness of the operators entering the
scene.

In particular, the latter constraint has been applied to the so
called metric operator $\Theta$ which enters the definition of the
inner product in the so called ``standard" physical Hilbert space of
states ${\cal H}^{(S)}$ (this notation has been introduced in
\cite{SIGMA}). Under such a constraint we followed the notation
conventions introduced in the series of recent papers by one of us
(F.B.) and spoke about the ``regular" NLPB systems in paper I.

In this context, our present paper II will start from an appropriate
weakening of the assumptions. This will enable us to formulate,
rigorously, the third, ``missing" connection between the NLPB
systems and NLSUSY algebras.

Our constructions will start from a systematic clarification of the
appropriate definitions. Firstly, the notion of the cryptohermitian
Hamiltonians will be left reserved for the class of bounded
operators $H\neq H^\dagger$. The phenomenologically inspired
consistency of the use of such a severely restricted class has been
advocated by Scholtz et al \cite{Geyer} who imagined that the
related simplification of the mathematics proves vital, in their
case of interest, for the practical feasibility of the
interacting-boson-model-inspired variational calculations of the
spectra of the heavy nuclei.

In the present context motivated by the needs of supersymmetry, the
overall situation is much less easy. First of all, one cannot
restrict one's attention to the bounded (i.e., in our notation,
cryptohermitian) Hamiltonians $H\neq H^\dagger$ anymore. In order to
reflect such an important change of perspective, we shall rechristen
the ``unbounded cryptohermitian" Hamiltonians $H\neq H^\dagger$ to
``quasi-Hermitian" Hamiltonians. Such a terminological aspect of the
problem  has also been discussed, after all, also in the
introductory part of our preceding paper I. In the present paper
such a terminological convention may find an independent and very
sound historical support in the introduction of such a name, by
Dieudonn\'{e} \cite{Dieudonne}, as early as in 1964.

Within the broadened perspective, the present usage of the name of
quasi-Hermitian Hamiltonians will be mostly accompanied by the
concrete selection of an ordinary differential linear Hamiltonian,
like the ${\cal PT}-$symmetric Hamiltonians cited above. Let us
remind the readers that  we have shown in paper I that the notions
of {\em regular non-linear pseudo-bosons} and {\em
cryptohermiticity} are, under certain sound assumptions, equivalent.
One of the assumptions used all along that paper is related to the
fact that the intertwining operator is bounded with bounded inverse
or, equivalently, that the two sets of eigenstates of $M$ and
$M^\dagger$ are Riesz bases. However, in the above-mentioned
physical applications (and many other ones)  this is not ensured at
all. In these cases the role of the unbounded operators becomes
crucial.

In the present paper we shall show that many of our previous results
can still be extended when the unbounded operators are involved. The
paper is organized as follows: in section \ref{II}  we shall return
to the notion of the ``hidden" Hermiticity \cite{SIGMA} and
distinguish, for our present purposes at least, between its form
called cryptohermiticity (in which one assumes that the operators
are bounded) and its generalized, unbounded-operator form which will
be called here, for the sake of definiteness, {\em
quasi-Hermiticity}. Subsequently we return to the definition of {\em
non-linear pseudo-bosons} (NLPB) and focus on the case in which
these cease to be regular. In such a setting we shall outline
parallels as well as differences between the results of paper I.
Section \ref{III} is then devoted to examples, while our conclusions
are given in Section \ref{IV}.

\section{Observables and metrics: bounded versus unbounded\label{II}}

\subsection{Cryptohermiticity versus quasi-Hermiticity}

Let us, once more, return to the above-mentioned unification of
bosons with fermions and recall the popular idea of their
arrangement into the so called supersymmetric multiplets. This idea
found a wide acceptance by the theoretical particle physicists
although, up to now, it does not seem supported by any persuasive
experimental evidence. This is the main ``hidden" reason why the
formalism has thoroughly been tested via the toy-model formalism of
the so called supersymmetric quantum mechanics (SUSYQM,
\cite{Witten}). The simplification proved suitable for the purpose.
For the sake of brevity one restricts there one's attention just to
a system composed of a single linear fermion in a combination with
an arbitrarily large $n-$plet of the linear bosons \cite{Khare}.

Fortunately, the subsequent study of SUSYQM found an independent and
fruitful motivation in its own, mostly purely formal byproducts.
{\em Pars pro toto} we might mention the development of the concept
of the shape invariance of solvable two-particle potentials, etc.

One of the other useful byproducts of the study of SUSYQM may be
seen, paradoxically, in it incompleteness as noticed by Jewicki and
Rodrigues \cite{JR}. On an abstract level this point may be
characterized as a sort of incompatibility between the analytic
implementation and the algebraic essence of the formalism. Indeed,
in the latter context one reveals that a {\em different}
angular-momentum-like parameter $\ell$ enters, in principle, the two
partner Hamiltonian-like operators via the centrifugal-like
interaction term $\sim \ell(\ell+1)/r^2$. In the former context, as
a consequence, one must very carefully discuss the boundary
conditions in the origin.

Fortunately, in the traditional SUSYQM of the textbooks, it is quite
easy to satisfy these $\ell-$dependent boundary conditions (and to
ignore the whole ``algebraic" shortcoming) by using simply a
brute-force suppression of the ``dangerous"  $\ell-$dependence of
the Hamiltonians in question. Roughly speaking, one simply decides
to restrict one's attention just to the special cases in which
$\ell(\ell+1)=0$ \cite{Bagchi}.

An unexpectedly successful alternative recipe of the extension of
the theory to all of the ``reasonable" real $\ell > -1/2$
(performing, in effect, its regularization) has been found in the
small-circle complexification of the coordinate $r$ near the origin
\cite{BG}. Such an origin-avoiding regularization of the
Schr\"{o}dinger equation breaks, naturally, the manifest Hermiticity
of the Hamiltonian and/or partner sub-Hamiltonians in question. For
this reason, one must be rather careful -- in our present paper we
shall return to the domain covered by the textbooks by using the
recipes as summarized rather briefly in Ref.~\cite{SIGMA} or in our
preceding paper I.

At this point it is important to emphasize that in the latter two
papers (as well as in their ``fathers-founders'" predecessor
\cite{Geyer}) the formalism of the so called ``cryptohermitian"
quantum mechanics is built upon the mathematics-simplifying
assumption that all of the operators entering the game are bounded.
We are now interested in discussing the mathematically more
sophisticated version of the formalism where the emphasis is being
shifted to the differential versions of the operators, with a number
of illustrative differential-equation examples as reviewed, say, in
long papers \cite{Carl,ali}.

For an incorporation of the related necessary weakening of the
assumptions let us first introduce the following

\begin{defn}\label{defnch}
Let us consider two operators $H$ and $\Theta$ acting on the Hilbert
space $\Hil$, with $\Theta$ self-adjoint,  positive and invertible.
Let us call $H^\dagger$ the adjoint of $H$ in $\Hil$ with respect to
its scalar product and introduce the conjugate operator
$H^\ddagger=\Theta^{-1}H^\dagger\Theta$, whenever it exists. We will
say that $H$ is quasi-Hermitian with respect to $\Theta$
(QHwrt$\Theta$) if $H=H^\ddagger$.
\end{defn}

\subsection{Quasi-Hermiticity versus the NLPB properties}

It is worth reminding the readers that we are interested in the case
in which $\Theta$ and $\Theta^{-1}$ are unbounded. Using standard
facts in functional calculus it is obvious that, in the assumptions
of Definition \ref{defnch} the operators $\Theta^{\pm 1/2}$ are well
defined. Hence we can introduce an operator
$h:=\Theta^{1/2}\,H\,\Theta^{-1/2}$, at least if the domains of the
operators allow us to do so. More explicitly, $h$ is well defined
if, taken $f\in D(\Theta^{-1/2})$, $\Theta^{-1/2}f\in D(H)$ and,
moreover, if $H\,\Theta^{-1/2}f\in D(\Theta^{1/2})$.

Of course, the latter requirements are surely satisfied if $H$ and
$\Theta^{\pm 1/2}$ are bounded. This option was considered in paper
I. Otherwise, due care is required, forcing us to introduce the
following, slightly modified terminology.

\begin{defn}\label{defnwellbeh}
Assume that  $H$ is QHwrt$\Theta$, for $H$ and $\Theta$ as above.
$H$ is {\em well behaved} wrt $\Theta$ if (i)
$h=\Theta^{1/2}\,H\,\Theta^{-1/2}$ exists and is self-adjoint,
$h=h^\dagger$; (ii) $h$ has only discrete eigenvalues $\epsilon_n$,
$n\in {\Bbb N}_0:={\Bbb N}\cup\{0\}$, with eigenvectors $e_n$:
$he_n=\epsilon_n e_n$, $n\in {\Bbb N}_0$, and (iii) if $\E:=\{e_n\}$
is an o.n. basis on $\Hil$.
\end{defn}

\vspace{2mm}

This definition is slightly different from that considered in paper
I, and it is more convenient in the present context where $\Theta$
is assumed to be unbounded. Similarly, the general notion of NLPB
should also incorporate the cases which are not regular.

\begin{defn}\label{defnlpb}
Given two operators $a$ and $b$ acting on Hilbert space $\Hil$ we
will say that the triple $(a,b,\{\epsilon_n\})$ such that
$\epsilon_0=0<\epsilon_1<\cdots<\epsilon_n<\cdots$ is a family of
NLPB if the following four properties hold:

\begin{itemize}

\item {\bf p1.}
a non zero vector $\Phi_0$ exists in $\Hil$ such that $a\,\Phi_0=0$
and $\Phi_0\in D^\infty(b)$.

\item {\bf { p2}.}
a non zero vector $\eta_0$ exists in $\Hil$ such that
$b^\dagger\,\eta_0=0$ and $\eta_0\in D^\infty(a^\dagger)$.

\item {\bf { p3}.} calling
 \be
  \Phi_n:=\frac{1}{\sqrt{\epsilon_n!}}\,b^n\,\Phi_0,\qquad
\eta_n:=\frac{1}{\sqrt{\epsilon_n!}}\,{a^\dagger}^n\,\eta_0,
\label{55}
 \en
we have, for all $n\geq0$, $\Phi_n\in D(a)$, $\eta_n\in
D(b^\dagger)$ and
 \be
 a\,\Phi_n=\sqrt{\epsilon_n}\,\Phi_{n-1},\qquad
b^\dagger\eta_n=\sqrt{\epsilon_n}\,\eta_{n-1}. \label{56}\en

\item {\bf { p4}.}
The sets $\F_\Phi=\{\Phi_n,\,n\geq0\}$ and
$\F_\eta=\{\eta_n,\,n\geq0\}$ are bases of $\Hil$.


\end{itemize}

\end{defn}

The definitions in (\ref{55}) are well posed in the sense that,
because of {\bf p1} and {\bf p2}, the vectors $\Phi_n$ and $\eta_n$
are well defined vectors of $\Hil$ for all $n\geq0$ \cite{bagnlpb}.
In paper I we further assumed that $\F_\Phi$ and $\F_\eta$ are Riesz
bases of $\Hil$. Under such a constraint we called our NLPB  {\em
regular} (NLRPB). Now, we will consider the fully general case in
which the latter condition is {\em not} satisfied. For the sake of
brevity of our discussion we shall, at the same time, skip the not
too interesting possibility of having the multiplicity
$m(\epsilon_n)$ of some eigenvalues $\epsilon_n$ greater than one.

Definition \ref{defnwellbeh} above will then imply that the set $\E$
produces a resolution of the identity which we write in the bra-ket
language as
 $$\sum_{n=0}^\infty |e_n\kt\,\br e_n|=\1.$$
Proceeding further in a close parallel with paper I let us now
introduce the manifestly not self-adjoint operators
 \be
  M=b\,a,\qquad \M=M^\dagger=a^\dagger b^\dagger.
 \label{57}
 \en
We can check that
 $$\Phi_n\in D(M)\cap D(b)\cap D(a)\,, \ \ \
 \eta_n\in D(\M)\cap D(a^\dagger)\cap D(b^\dagger),$$
as well as
 \be
 b\,\Phi_n=\sqrt{\epsilon_{n+1} }\,\Phi_{n+1},
 \qquad
 a^\dagger\eta_n=\sqrt{\epsilon_{n+1} }\,\eta_{n+1}, \label{58}
 \en
 \be
 M\Phi_n=\epsilon_n\Phi_n,\qquad \M\eta_n=\epsilon_n\eta_n
 \label{59}
 \en
which follow from definitions (\ref{55}) and from (\ref{56}).
Incidentally, this does not automatically imply that, for instance,
$D(a)$ is exactly the linear span of the $\Phi_n$'s, $D_\Phi$, but
only that  $D(a)\supseteq D_\Phi$. The eigenvalue equations
themselves imply that the vectors in $\F_\Phi$ and $\F_\eta$ are
mutually orthogonal,
 \be
 \left<\Phi_n,\eta_m\right>=\delta_{n,m}, \label{510}
 \en
having fixed the normalization of $\Phi_0$ and $\eta_0$ in such a
way that $\left<\Phi_0,\eta_0\right>=1$. Recalling \cite{bagnlpb} we
remind the readers that conditions  \{{\bf p1}, {\bf p2}, {\bf p3},
{\bf p4}\} are equivalent to  \{{\bf p1}, {\bf p2}, {\bf p3}$'$,
{\bf p4}\}, where

\vspace{2mm}

{\bf {\bf p3}$'$.} The vectors $\Phi_n$ and $\eta_n$ defined in
(\ref{55}) satisfy (\ref{510}).

\vspace{2mm}

 \noindent
Let us now complement $M$ and $\M$ by a pair of further operators
 \be
 N:=a\,b, \qquad \N:=N^\dagger=b^\dagger a^\dagger.
 \label{511}
 \en
It is easy to check that $\Phi_n\in D(N)$, $\eta_n\in D(\N)$, and
that $N\Phi_n=\epsilon_{n+1}\Phi_n$ and
$\N\eta_n=\epsilon_{n+1}\eta_n$, for all $n\geq0$. If the sequence
$\{\epsilon_n\}$ diverges for diverging $n$ it is clear that the
operators involved here, $a$, $b$, $M$, $N$ and so on, are
unbounded. Moreover, as already discussed in the Introduction, also
the intertwining operator between $M$ and $M^\dagger$, see below,
will turn out to be unbounded, contrarily to what happens in paper
I. For this reason we will pay particular attention to this aspect
of our construction.

To begin with, let us define an operator $X$ on a certain dense
domain $D(X)$ as follows: $D(X)\ni f\rightarrow
Xf:=\sum_{k=0}^\infty\left<\eta_n,f\right>\Phi_n$. The set $D(X)$
contains, for instance, all the vectors of $\F_\Phi$, whose linear
span is dense in $\Hil$: hence the norm closure of $D(X)$ is all of
$\Hil$, so that $X$ is well defined. Now, for all $f\in D(X)$ and
for all $m\geq0$, we see that $\left<\eta_m,(Xf-f)\right>=0$.
Therefore, since $\F_\eta$ is complete in $\Hil$, $Xf=f$. In other
words, $X$ is the identity operator on $D(x)$ and it can be extended
to all of $\Hil$. Then, using Dirac's bra-ket notation, we can write
 \be
 \sum_n|\Phi_n\kt\,\br \eta_n|=\sum_n|\eta_n\kt\,\br \Phi_n|=\1.
 \label{512}
 \en
Let us now define two more operators, $S_\Phi$ and $S_\eta$, on
their domains $D(S_\Phi)$ and $D(S_\eta)$, by letting
$h=\sum_n\left<\Phi_n,h\right>\eta_n$ be in $D(S_\Phi)$ and setting
 \be S_\Phi h=
 \sum_n\left<\Phi_n,h\right>\Phi_n.\label{515}
 \en
Analogously, let $f=\sum_n\left<\eta_n,f\right>\Phi_n$ be in
$D(S_\eta)$. Then we define
 \be S_\eta f=
 \sum_n\left<\eta_n,f\right>\eta_n.\label{516}
 \en
In the Dirac's notation this means that
$S_\Phi:=\sum_n|\Phi_n\kt\,\br \Phi_n|$ and
$S_\eta:=\sum_n|\eta_n\kt\,\br \eta_n|$. It is clear that both these
operators are densely defined. Indeed, calling as before $D_\Phi$
and $D_\eta$ respectively the linear spans of $\F_\Phi$ and
$\F_\eta$, we see that $D_\Phi\subseteq D(S_\eta)$ and
$D_\eta\subseteq D(S_\Phi)$. In particular,
 \be
 S_\eta \Phi_n=\eta_n, \qquad S_\Phi \eta_n=\Phi_n, \label{517}
 \en
for all $n\geq0$. The last equations have an interesting
consequence: since  $\F_\Phi$ and $\F_\eta$ are not Riesz
bases\footnote{Recall that this is the situation we are interested
in, here. The case in which these are Riesz bases was already
considered in paper I}, $S_\eta$ and $S_\Phi$ are necessarily
unbounded operators. This means that they cannot be considered, or
called, {\em frame operators}, as in our previous papers, since in
standard frame theory the frame operator is necessarily bounded with
bounded inverse. It is also easy to check that they are both
positive definite, $\left<h, S_\Phi h\right>>0$, $\left<f, S_\eta
f\right>>0$, for all non zero $h\in D(S_\Phi)$ and $f\in D(S_\eta)$,
and that they are one the inverse of the other:
 \be
 S_\eta=S^{-1}_\Phi.\label{518}
 \en
For that, we
have to prove that, if $f\in D(S_\eta)$, then $S_\eta f\in
D(S_\Phi)$ and $S_\Phi(S_\eta f)=f$, and that, if $h\in D(S_\Phi)$,
then $S_\Phi h\in D(S_\eta)$ and $S_\eta(S_\Phi h)=h$.

Let $f\in D(S_\eta)$ be a norm-limit $f=\|.\|-\lim_{N \to\infty}f_N$
of $f_N=\sum_{k=0}^N\left<\eta_k,f\right>\Phi_k$, with $S_\eta f_N$
converging uniformly in $\Hil$ to what we call $(S_\eta f)$. In
other words, both $\{f_N\}$ and $\{S_\eta f_N\}$ are $\|.\|$-Cauchy
sequences. To check that $(S_\eta f)$ belongs to $D(S_\Phi)$ it is
enough to check that $\{S_\Phi(S_\eta f_N)\}$ is a $\|.\|$-Cauchy
sequence as well. This is true since $S_\Phi(S_\eta f_N)=f_N$ for
all $N$, which is a  $\|.\|$-Cauchy sequence by assumption,
converging to $f$. This concludes half of what we had to prove. The
proof of the other half is similar.

A direct computation finally shows that
$D(S_\eta^\dagger)=D(S_\eta)$, $D(S_\Phi^\dagger)=D(S_\Phi)$,  and
that $S_\eta=S_\eta^\dagger$ and $S_\Phi=S_\Phi^\dagger$.

\vspace{2mm}

{\bf Remark:} An apparently simpler definition of $S_\eta$ and
$S_\Phi$ would consist in fixing their domains to be exactly
$D_\Phi$ and $D_\eta$, respectively. This is equivalent to a
restriction of the operators considered so far. However, this choice
is not appropriate for us since, in particular, it is not clear if
for instance $D(S_\Phi)=D(S_\Phi^\dagger)$ \cite{Dieudonne}.
Nevertheless, similar restrictions will be quite useful in the next
section.

\subsection{Relating $M$ and $M^\dagger$ for non-regular NLPB}

We are now interested in deducing a relation between $M$ and
$M^\dagger$ using the operators $S_\Phi$ and $S_\eta$. The starting
point is the eigenvalue equation $M\Phi_n=\epsilon_n\Phi_n$,
together with the equality $\eta_n=S_\eta\Phi_n$ obtained before.
Hence $M\Phi_n\in D(S_\eta)$ and we have that
$S_\eta(M\Phi_n)=\epsilon_n\eta_n$, for all $n\geq0$. This equation
implies also that $\eta_n\in D(S_\eta MS_\Phi)$, and that, for all
$n\geq0$,
 \be
  \left(S_\eta MS_\Phi-\M\right)\eta_n=0.\label{519}
  \en
This equation, by itself, is not enough to ensure that $S_\eta
M\Phi_n=\M$. We know (see \cite{hal}, Problem 50) that for an
unbounded operator $A$, the validity of equation $Ae_n=0$ for all
vectors $e_n$ of a basis still does not imply, in general, that
$A=0$\footnote{In order to be so, we should have $Ae_n=0$ for all
bases!}. In other words, even if it is rather reasonable to imagine
that $\left(S_\eta M\Phi_n-\M\right)\eta_n=0$ implies that $S_\eta
M\Phi_n=\M$, this is not guaranteed at all. For this reason, as
already anticipated in the previous remark, we define the following
restrictions:
 \be
 M_0=M\upharpoonright_{D_\Phi}, \quad
N_0=N\upharpoonright_{D_\Phi},\quad \M_0=\M\upharpoonright_{D_\eta},
\quad \N_0=\N\upharpoonright_{D_\eta}. \label{521}
 \en
For these operators we can prove that
 \be
 S_\eta M_0 S_\Phi=\M_0,
\qquad M_0^\dagger=\M_0,\label{522}
 \en
as well as
 \be S_\eta N_0 S_\Phi=\N_0, \qquad
N_0^\dagger=\N_0.\label{523}
 \en
Indeed we can check that, for instance, $D(\M_0)=D(S_\eta M_0
S_\Phi)=D_\eta$ and that the operators $\M_0$ and $S_\eta M_0
S_\Phi$ coincide on $D_\eta$. Therefore, for these restrictions,
formulas analogous to those found in paper I are recovered.


The following theorem, which extends to non regular NLPB an
analogous result proven in paper I, can now be deduced:

\begin{thm}
Let  $H$ be well behaved wrt $\Theta$, with $\Theta=\Theta^\dagger$
unbounded, positive and invertible. Then it is possible to introduce
two operators $a$ and $b$ on $\Hil$, and a sequence of real numbers
$\{\epsilon_n,\,n\in {\Bbb N}_0 \}$, such that the triple
$(a,b,\{\epsilon_n\})$ is a family of non regular NLPB.

Vice-versa, if  $(a,b,\{\epsilon_n\})$ is a family of non regular
NLRB, two operators can be introduced, $H$ and $\Theta$, such that
$\Theta=\Theta^\dagger$ is unbounded, positive and invertible, and
$H$ is well behaved wrt $\Theta$.
\end{thm}

\begin{proof}

The proof is slightly different from that given for bounded
operator, so that we will give it here.

First, we assume that $H$ is
well behaved wrt $\Theta$, where $\Theta=\Theta^\dagger$ is an
unbounded, positive and invertible operator. Of course, our
hypotheses imply that (i) $H^\ddagger:=\Theta^{-1}H^\dagger\Theta$
is well defined and coincides with $H$; (ii) that
$h=\Theta^{1/2}H\Theta^{-1/2}$ is also well defined, and
self-adjoint; (iii) that $\E$ is an o.n. basis of eigenvectors of
$h$, with eigenvalues $\{\epsilon_n\}$, of $\Hil$: $he_n=\epsilon_n
e_n$, for all $n\geq0$.

Therefore, $\Theta^{1/2}H\Theta^{-1/2}e_n=\epsilon_n e_n$, $e_n\in
D(\Theta^{-1/2})$; consequently, $H(\Theta^{-1/2}e_n)=\epsilon_n
(\Theta^{-1/2}e_n)$. This suggests to define the vectors
$\Phi_n:=\Theta^{-1/2}e_n$, which belong to $D(H)$ and satisfy the
eigenvalue equation $H\Phi_n=\epsilon_n\Phi_n$. Since $h=h^\dagger$,
we can  repeat the same considerations starting from $h^\dagger$.
Hence, defining $\eta_n:=\Theta^{1/2}e_n$, we deduce that $\eta_n\in
D(H^\dagger)$ and that $H^\dagger\eta_n=\epsilon_n\eta_n$. The sets
$\F_\Phi:=\{\Phi_n,\,n\geq0\}$ and $\F_\eta:=\{\eta_n,\,n\geq0\}$
can be proven to be bases of $\Hil$. Indeed, let us take a vector
$f\in D(\Theta^{1/2})$, such that $f$ is orthogonal to all the
vectors in $\F_\eta$. Therefore we have, for all $n\geq0$,
 $$
0=\left<f,\eta_n\right>=\left<f,\Theta^{1/2}e_n\right>
=\left<\Theta^{1/2}f,e_n\right>,
 $$
which implies that $\Theta^{1/2}f=0$, so that $f=0$ as well. Using
standard results, see \cite{han} for instance, we conclude that all the
elements of $\Hil$ can be expanded in terms of $\F_\eta$, which is
therefore a basis of all of $\Hil$. Analogously, we can check that
$\F_\Phi$ is a basis of $\Hil$. However, due to the fact that
$\Theta^{\pm1/2}$ are unbounded, $\F_\eta$ and $\F_\Phi$ are not
Riesz bases.

Let us now define two operators $a$ and $b$ on $D(a)=D(b):=D_\Phi$
as follows: let $f=\sum_{k=0}^Nc_k\Phi_k$ be a generic vector in
$D_\Phi$. Then
 \be
a\,f:=\sum_{k=1}^Nc_k\sqrt{\epsilon_k}\Phi_{k-1},\qquad
b\,f:=\sum_{k=0}^Nc_k\sqrt{\epsilon_{k+1}}\Phi_{k+1}.
 \label{525}
 \en
In particular these imply that
$a\,\Phi_n=\sqrt{\epsilon_n}\,\Phi_{n-1}$ and that
$b\,\Phi_n=\sqrt{\epsilon_{n+1}}\,\Phi_{n+1}$, for all $n\geq0$.
Now, recalling that $\epsilon_0=0$, we deduce that $a\Phi_0=0$.
Also, iterating the raising equation above, we find that
$\Phi_n:=\frac{1}{\sqrt{\epsilon_n!}}\,b^n\,\Phi_0$, which implies,
in particular, that $\Phi_0\in D^\infty(b)$. Hence condition {\bf
p1} of Definition \ref{defnlpb} is satisfied.

To check condition {\bf p2} we first have to compute $a^\dagger$ and
$b^\dagger$. It is possible to check that, for all $n\geq0$,
$\eta_n\in D(a^\dagger)\cap D(b^\dagger)$, and that \be
a^\dagger\eta_n=\sqrt{\epsilon_{n+1}}\,\eta_{n+1},\qquad
b^\dagger\eta_n=\sqrt{\epsilon_{n}}\,\eta_{n-1}, \label{526}\en so
that, clearly, $b^\dagger\eta_0=0$ and, again acting iteratively,
$\eta_n\in D^\infty(a^\dagger)$. In fact, we find that
$\eta_n:=\frac{1}{\sqrt{\epsilon_n!}}\,{a^\dagger}^n\,\eta_0$.
Condition {\bf p3$'$} is clearly true, while condition {\bf p4} was
already proved.

\vspace{2mm}

Let us now prove the converse implication, that is, let us see how
NLPB  produce two operators, $H$ and $\Theta$, satisfying Definition
\ref{defnwellbeh}.

This is a consequence of equation (\ref{522}), $S_\eta M_0
S_\Phi=\M_0$, which we can rewrite as $M_0=S_\eta^{-1}M_0^\dagger
S_\eta$. Hence, the operators $H$ and $\Theta$ in Definition
\ref{defnch} are easily identified: $H$ is $M_0$, while $\Theta$ is
$S_\eta$, and $M_0$ is QHwrt$S_\eta$. With this in mind the operator
$h$ becomes $h=S_\eta^{1/2}M_0S_\eta^{-1/2}$. First of all, we need
to understand if $h$ is well defined. For that, recalling the
properties of $S_\eta$ and using the spectral theorem, we deduce
that $S_\eta^{\pm 1/2}$ are well defined.

Let us now observe that, if $f\in D(S_\eta)$, then $f\in
D(S_\eta^{1/2})$. This follows from the equality $\left<f,S_\eta
f\right>=\|S_\eta^{1/2}f\|^2$. Analogously, if $h\in
D(S_\eta^{-1})$, then $h\in D(S_\eta^{-1/2})$. Therefore, since
$\Phi_n\in D(S_\eta)$, $\Phi_n\in D(S_\eta^{1/2})$ as well, so that
we can define new vectors of $\Hil$ as $e_n:=S_\eta^{1/2}\Phi_n$,
$n\geq0$. Notice that $e_n\in D(S_\eta^{1/2})\cap D(S_\eta^{-1/2})$.
In fact we have: $S_\eta^{1/2}e_n=S_\eta \Phi_n=\eta_n$, and
$S_\eta^{-1/2}e_n=\Phi_n$. It follows that $e_n\in D(h)$ and that
$he_n=\epsilon_ne_n$. Standard arguments, \cite{han}, finally show
that the linear span of $\E:=\{e_n\}$ is dense in $\Hil$, showing in
this way that $h$ is well defined. Finally, we can also check from
the definition that $\left<e_n,e_m\right>=\delta_{n,m}$: $\E$ is an
o.n. basis of $\Hil$. It is now clear that $h=h^\dagger$.

\end{proof}

We want to briefly consider few consequence of this theorem, which
are very similar to those found in paper I.
\begin{enumerate}

\item
The Dirac's representations of the operators introduced so far can
again be easily deduced. Thus, we have
 \be a=\sum_{n=0}^\infty
\sqrt{\epsilon_n}|\Phi_{n-1}\kt\,\br \eta_n|,\quad
b=\sum_{n=0}^\infty \sqrt{\epsilon_{n+1}}|\Phi_{n+1}\kt\,\br
\eta_n|. \label{expr}
 \ee
We can also deduce the similar expansions for $a^\dagger$ and
$b^\dagger$ and for
$$h=\sum_{n=0}^\infty \epsilon_n|e_{n}\kt\,\br e_n|,\quad H
=\sum_{n=0}^\infty \epsilon_n|\Phi_{n}\kt\,\br \eta_n|, \mbox{ and }
H^\dagger=\sum_{n=0}^\infty \epsilon_n|\eta_{n}\kt\,\br \Phi_n|.$$
 \item
As in paper I, operators $S_\eta$ and $S_\Phi$, and their square
roots, behave as intertwining operators. This is exactly the same
kind of result we have deduced for {\em regular} pseudo-bosons,
where biorthogonal Riesz bases and intertwining operators are
recovered. For instance, equation (\ref{522}) produces the following
intertwining relation: $S_\Phi M_0=\M_0 S_\Phi$.
 \item
Even if $h$ is not required to be factorizable, it turns out that it
can still be written as $h=b_\Theta a_\Theta$, where
$a_\Theta=\Theta^{1/2}a\Theta^{-1/2}$ and
$b_\Theta=\Theta^{1/2}b\Theta^{-1/2}$. We can write
$[a_\Theta,b_\Theta]=\Theta^{1/2}[a,b]\Theta^{-1/2}\neq [a,b]$, but
if $\left[[a,b],\Theta^{1/2}\right]=0$, which is the case for linear
pseudo-bosons. Thus, hamiltonian $h$ can be written in a factorized
form, at a formal level at least.

\end{enumerate}

\section{Non-regular NLPB in differential-operator realizations \label{III}}

This section will be divided in two parts, with the first one
offering a physical motivation and background of what will be
discussed in the second subsection.

\subsection{Nonlinear supersymmetries}

\subsubsection{Antilinear operators}

In a historical perspective and in the context of physics and
quantum theory the emergence of the pair of non-selfadjoint
factorized operators (\ref{57}) may be traced back to
Ref.~\cite{z1ZCBR}. In this letter the usual form of  supersymmetric
quantum mechanics (in which one traditionally assumes that $M =
M^\dagger$ \cite{Khare}) has been generalized. In our present
language the idea of Ref.~\cite{z1ZCBR} (cf. also its presentation
in a broader context in Ref.~\cite{z3DI}) may be characterized as
lying in the use of nonlinear regular pseudo-bosons. Indeed, in the
approach of Ref.~\cite{z1ZCBR} using $M \neq M^\dagger$ the
supersymmetry connecting bosons with fermions has been realized in
the representation space spanned by states defined by
Eq.~(\ref{59}).

The quantum system presented in Ref.~\cite{z1ZCBR} may be recalled
here as our first illustration of the immediate applicability of the
general NLRPB formalism in the very concrete physical and
phenomenologically oriented situations. Firstly, following the
notation of Ref.~\cite{z1ZCBR} we have to define the {\em pair} of
the factorized sub-Hamiltonian operators $M=
M^{(\pm)}=B^{(\pm)}A^{(\pm)}$ where the quadruplet of the factors
may be chosen in the form of linear differential operators
 \be
A^{(\pm)} = \frac{d}{dx} + W^{(\pm)}(x), \ \ \ \ \ \ \ \ \ \ \
B^{(\pm)} =- \frac{d}{dx} + W^{(\pm)}(x)\,. \label{creanih}
 \ee
Once we fix a real constant $\varepsilon>0$ and select, for the sake
of definiteness,
 \be
  W^{(\pm)}(x) = \pm \left [
 \frac{1}{x\pm i\varepsilon}-
i\,(x\pm i\varepsilon)^2 \right ]\  \label{susyp}
 \ee
the main result is the validity of the refactorization
 \be
M^{(+)} = {\cal T} A^{(-)} B^{(-)} {\cal T} \  .\label{promis}
 \ee
The complex-conjugation antilinear operator ${\cal T}$ can be
interpreted as mimicking the time-reversal operation performed over
the system.

The readers are recommended to find more details (e.g., the relevant
older references and/or a generalization of the ansatz (\ref{susyp})
in {\em loc. cit.}). It is worth adding that the transition to
non-hermitian interactions makes the model truly inspiring. Its
structure may be perceived as an immediate predecessor of the
introduction of the abstract concept of pseudo-bosons in
Ref.~\cite{tri} where also an immediate follow-up preprint \cite{z2}
has been cited.

A few years later a further, so called tobogganic generalization of
the whole formalism has been proposed and summarized, say, in the
recent compact review paper \cite{z4acta}. The core of the
generalization lied in the Riemann-surface-adapted generalization of
the operator ${\cal T} \neq {\cal T}^{-1} $. Due to the
circumstances one must set
 \ben
  A=
 -{\cal T}\frac{d}{dx} + {\cal T}W^{(-)}(x)\,,\ \ \
 B=\frac{d}{dx}{\cal T}^{-1} + W^{(-)}(x){\cal T}^{-1}
 \,
 \label{reredefinition}
 \een
i.e., one must redefine further the creation- and annihilation-like
operators of Eq.~(\ref{creanih}).

\subsubsection{Regularizations
by complexifications}

Among the illustrative textbook quantum systems of SUSYQM a special
role is played by the one-dimensional harmonic-oscillator
Schr\"{o}dinger equation
 \be
 \left (-\,\frac{d^2}{dr^2} + r^2 \right )
 \, \psi(r) =
 E  \, \psi(r), \ \ \ \ \ \psi(r) \in L_2(-\infty,\infty)\,.
 \label{regho}
 \ee
In Ref.~\cite{ptho} this example found a natural PT-symmetric
two-parametric generalization in the so called Kratzer's harmonic
oscillator
 \be
  \left (-\,\frac{d^2}{dx^2}
+ \frac{G}{(x-ic)^2} + x^2 -2ic\,x - c^2\right )
 \, \varphi(x) =
E  \, \varphi(x), \ \ \ \ \ \varphi(x) \in L_2(-\infty,\infty).
 \label{SE}
 \ee
Here the real constant $c\neq 0$ regularizes the centrifugal-like
spike at any coupling strength $G = \alpha^2-1/4$ so that the wave
functions may be defined as living on the whole real line. The
parameter $\alpha$ should be chosen positive and, in the simpler,
non-degenerate case, non-integer. This implies~\cite{ptho} that the
complete set of normalizable eigenfunctions may be numbered by the
quasi-parity $q = \pm 1 $ and by the excitation  $ n = 0, 1, 2,
\ldots$. At the respective $c-$independent bound-state energies
 \be
 E=E_{qn}=4n+2 + 2 q \alpha\,
 \label{55a}
 \ee
wave functions become defined, in closed form, in terms of Laguerre
polynomials $L^{(\gamma)}_n (z)$,
  \be
\varphi(x) = const. \,(x-ic)^{-q \alpha+1/2}e^{-(x-ic)^2/2} \ L^{(-q
\alpha)}_n \left [
 (x-ic)^2
 \right ]\,.
\label{waves}
 \ee
Naturally, the new spectrum of energies is not equidistant, though
it is still real and composed of the two equidistant subspectra.

Although the reality of the energies (\ref{55a}) of the states
(\ref{55}) themselves (possessing, in addition, the so called
unbroken PT symmetry \cite{BB}) seemed to be in contrast with the
manifest non-Hermiticity of the underlying operators $M$, the puzzle
has been clarified in Ref.~\cite{z4ps}. We were able to show there
that our apparently non-Hermitian model (\ref{SE}) generating the
real spectrum of energies  may be re-interpreted as self-adjoint.
For this purpose we showed, in~\cite{z4ps}, that the inner product
may be modified in such a way that the induced norm remains positive
definite. We also showed that in spite of the immanent ambiguity of
such ``hidden-Hermiticity-mediating" changes of the inner product,
one of the most natural definitions of a unique inner product may be
based on the use of ``quasi-parity" \cite{z4ps} (which is now better
known as ``charge" \cite{Carl}).

\subsubsection{The implementation of supersymmetry}

In the ultimate stage of development of the SUSYQM construction as
presented in Ref.~\cite{z5} we were able to describe one of the most
natural deformations of the structure of the creation- and
annihilation-operator algebra. Its detailed form followed from the
$c \neq 0$ regularization of the singular harmonic oscillator of
Eq.~(\ref{SE}) where the regularized Hamiltonian will be denoted by
the superscripted bracket symbol $ {H}_{}^{[\alpha]}$ in what
follows.

Our construction just paralleled the standard supersymmetrization of
the current, regular harmonic oscillator (\ref{regho}) (cf.
\cite{Khare} for details).  Firstly, we replaced the above-proposed
cubic-oscillator toy-model superpotential of Eq.~(\ref{susyp}) by
its harmonic-oscillator alternative
 \be
  W^{(\gamma)}(x) =
 x- ic
- \frac{\gamma+1/2}{x- ic}\,,\ \ \ \ \ c > 0\,
  \label{susypsi}
 \ee
with any real $\gamma$. Secondly, in the manner compatible with the
supersymmetric recipe yielding the two (viz., ``left" and ``right",
$\gamma-$numbered) families of quantum Hamiltonians
 \be
 H_{(L)}^{(\gamma)}=B\, A=\hat{p}^2+W^2-W'
 , \ \ \ \ \ \ \ \
 H_{(R)}^{(\gamma)}=A \, B=\hat{p}^2+W^2+W'
 \label{partner}
 \ee
we verified that
 \be
 {H}_{(L)}^{(\gamma)} = {H}_{}^{[\alpha]} -2\gamma-2,
 \ \ \ \ \ \
 {H}_{(R)}^{(\gamma)} = {H}_{}^{[\beta]} -2\gamma
 \label{Mtt}
 \ee
where ${\alpha}=|\gamma|$ and $\beta=|\gamma+1|$, respectively.
Further details may be found in Refs.~\cite{z5} and \cite{z6}.

\subsection{Non-regular pseudo-bosons}

\subsubsection{The hidden Lie-algebraic structures}

In ref. \cite{z5} we revealed the existence of the {\em
second-order} differential operators
 \ben
A^{(-\gamma-1)} \, A^{(\gamma)} = A^{(\gamma-1)} \, A^{(-\gamma)}
= {\bf A}(\alpha) \label{operatorsa}
 \een
 \ben
B^{(-\gamma)} \, B^{(\gamma-1)}= B^{(\gamma)} \, B^{(-\gamma-1)}
= {\bf B} (\alpha) \label{operatorsb}
 \een
which acted as the true respective annihilation and creation
operators in our spiked and complex harmonic-oscillator model,
 \ben
 {\bf A}(\alpha) \,
 {\cal L}^{(\gamma)}_{N+1}{}=c_5(N,\gamma)\,
 {\cal L}^{(\gamma)}_{N}{}
,\ \ \ \ \ \  {\bf B}(\alpha) \,
 {\cal L}^{(\gamma)}_{N}{}=c_5(N,\gamma)\,
 {\cal L}^{(\gamma)}_{N+1}{}
 \een
where $ c_5(N,\gamma)=-4\sqrt{(N+1)(N+\gamma+1)}$. The corresponding
generalization of the pseudobosonic version of the Heisenberg
algebra has been shown, in Ref.~\cite{z7pariz},  for the Lie algebra
$sl(2,I\!\!R)$ with the renormalized generators $ {\bf
A}(\alpha)/\sqrt{32}$, ${\bf B}(\alpha)/\sqrt{32}$ and
$H^{(\alpha)}/4$ and with the commutators
 \ben
  {\bf A}(\alpha) \,
   {\bf B}(\alpha) \,-\, {\bf B}(\alpha) \,
    {\bf A}(\alpha) =8\,H^{[\alpha]}
    \een
    and
 \ben
  {\bf A}(\alpha) \,H^{(\alpha)}\,-
   \,H^{(\alpha)}\, {\bf A}(\alpha)
   \,\equiv\,4\, {\bf A}(\alpha), \ \ \
 H^{(\alpha)}\,{\bf B}(\alpha)\, -\,
   {\bf B}(\alpha) \,H^{(\alpha)} \,\equiv\,4\, {\bf B}(\alpha)\,.
    \een


\subsubsection{Reinterpretation}

The operators ${\bf A}(\alpha)$ and ${\bf B}(\alpha)$, and the
functions ${\cal L}^{(\gamma)}_{N}$, allow us to construct a non
trivial example of NLPB satisfying {\bf p1}-{\bf p4} of Definition
\ref{defnlpb}. For that we begin to define two operators $a:=-{\bf
A}(\alpha)$ and $b=-{\bf B}(\alpha)$, a countable family of vectors
$\Phi_n:={\cal L}^{(\gamma)}_{n}$ and the following sequence of
non-negative numbers: $\epsilon_n=c_5(n+1,\gamma)^2$. Then
$a\Phi_n=\sqrt{\epsilon_n}\Phi_{n-1}$ and
$b\Phi_n=\sqrt{\epsilon_{n+1}}\Phi_{n+1}$. Let now $\hat\Hil$ be the
closure of the linear span of the vectors $\Phi_n$'s, which, in
general, is a proper subset of $\Hil$. $\hat\Hil$ is a Hilbert
space, in which a unique biorthogonal basis $\F_\eta=\{\eta_n\}$,
$\left<\eta_n,\Phi_m\right>=\delta_{n,m}$, can be introduced,
\cite{chri}. The first vector of this biorthogonal set, $\eta_0$,
satisfies condition {\bf p2} of Definition \ref{defnlpb}. Indeed we
have
 $$
 \left<b^\dagger\eta_0,\Phi_k\right>=\left<\eta_0,b\,\Phi_k\right>
 =\sqrt{\epsilon_{k+1}}\left<\eta_0,\Phi_{k+1}\right>=0,
 $$
for all $k\geq0$. Hence, being $\F_\Phi$ complete,
$b^\dagger\eta_0=0$. To check now that $\eta_0$ belongs to
$D^\infty(a^\dagger)$, we consider the following scalar product:
$\left<a^k\Phi_n,\eta_0\right>$, which is zero whenever $k>n$ due to
the lowering property of $a$ on the set $\F_\Phi$. On the other
hand, if $k\leq n$, we deduce that
$\left<a^k\Phi_n,\eta_0\right>=\sqrt{\epsilon_k!}\delta_{n,k}$.
Therefore, since $\delta_{n,k}=\left<\Phi_n,\eta_k\right>$, we
deduce that
 $\left<\Phi_n,\left(\frac{1}{\sqrt{\epsilon_k!}}\,
 (a^\dagger)^k\eta_0-\eta_k\right)\right>=0$
for all $n\geq0$. Hence, using once more the completeness of
$\F_\Phi$, we deduce that
$\eta_k=\sqrt{\epsilon_k!}\,(a^\dagger)^k\eta_0$ for all $k$: this
shows that $\eta_0\in D^\infty(a^\dagger)$ and that the various
vectors of the unique biorthogonal basis $\F_\eta$ introduced above
are related as in equation (\ref{55}). To fulfill all the
requirements of Definition \ref{defnlpb} we finally have to prove
that $b^\dagger\eta_n=\sqrt{\epsilon_n}\eta_{n-1}$. The proof goes
like this: for all $k\geq0$,
 $$
 \left<b^\dagger\eta_n,\Phi_k\right>=\left<\eta_n,b\,\Phi_k\right>
 =\sqrt{\epsilon_n}\delta_{n-1,k}
 =\sqrt{\epsilon_n}\,\left<\eta_{n-1},\Phi_k\right>,
 $$
so that
 $\left<\left(b^\dagger\eta_n
 -\sqrt{\epsilon_n}\eta_{n-1}\right),\Phi_k\right>=0$
for all $k$. Since $\F_\Phi$ is complete, this proves that $b$ is a
lowering operator for $\F_\eta$, as required. We conclude that NLPB
can be constructed out of this model, but they are not, in general,
regular, due to the fact that the operator $S_\eta$ mapping
$\F_\Phi$ to $\F_\eta$ is, in general, unbounded. Incidentally, we
also deduce that $\Phi_n$ and $\eta_n$ are respectively eigenstates
of $H^{(\alpha)}$ and ${H^{(\alpha)}}^\dagger$ with the same
eigenvalue, $\frac{1}{8}(\epsilon_{n+1}-\epsilon_n)$. This is
connected to the fact that these operators are related to the
operators $M=ba$ and $M^\dagger=a^\dagger b^\dagger$ (and to their
{\em specular} counterparts $N=ab$ and $N^\dagger=b^\dagger
a^\dagger$) introduced in Section II.

\section{Conclusions \label{IV}}

The key motivation of our present NLPB-related studies I and II was
twofold. Firstly, in a series of papers \cite{bagrev} - \cite{abg}
one of us (F.B.) considered the canonical commutation relations
$[a,b]=\1$ in a generalized version in which $b$ was not necessarily
equal to $a^\dagger$. In parallel, in another series of papers (cf.,
e.g., their most recent samples \cite{chebypol} - \cite{laguerre}),
the second one of us (M.Z.) studied the possibility of the weakening
of the Hermiticity of the observables in a few quantum systems of an
immediate phenomenological appeal and/or methodical interest.

In paper I we announced the possibility of connecting these two
alternative points of view. In particular we addressed the problem
while simplifying its technical aspects by the acceptance of the
operator-boundedness assumptions as currently made in the physics
literature \cite{Geyer}. This enabled us to clarify the role of the
metric (specifying the inner products in the correct Hilbert space
of states) from the NLPB point of view, and {\em vice versa}. We
also endorsed the message of Refs. \cite{Geyer} and \cite{SIGMA} by
re-recommending the practical use of the factorizations of the
metrics $\Theta$ into the individual Dyson-map factors $\Omega$.

Later on we consulted several less accessible mathematics-oriented
references (e.g., \cite{szym}) and imagined that there exist many
situations in physics (with some of them being cited above) in which
the picture provided by the bounded operators appears insufficient.
For this reason we returned to the subject of paper I. In its
present continuation we incorporated the above-mentioned new
knowledge and perspective into a necessary weakening of the
underlying mathematical assumptions.

In paper II we revealed, first of all, that in the territory of
unbounded operators the functional structure obtained from $a$ and
$b$, and from the so-called pseudo-bosons related to these, may be
much richer than the one described in paper I. Still, many ideas of
paper I survived the generalization. In particular, in the presently
specified {\em quasi-Hermitian} case we still succeeded in the
clarification of the conditions under which one can still work with
the NLPB formalism where the two biorthogonal bases remain
obtainable as eigenstates of the two {\em number-like} operators,
$M$ and $M^\dagger$, with eigenvalues which are not equal to
integers in general.

In a certain synthesis of our originally separate starting positions
we showed that even in the quasi-Hermitian context with unbounded
operators the doublet of the generalized number operators $M$ and
$M^\dagger$ (and of $N$ and $N^\dagger$ as well) may still be
perceived as interconnected by an intertwining operator. The latter
intertwinner has been shown specified using the two sets of
eigenstates. At this point we made an ample use of the extended
notion of pseudo-bosons in which their essential characteristics
play the role. In this sense we believe that the role of the
generalized number operators might acquire more and more relevance
in the future applications of the formalism where the boundedness of
the operators of the observables cannot be guaranteed and where, in
addition, the Hermiticity of these operators is ``hidden".

We dare to believe that our present results might encourage a
further growth of interest in the practical use of quasi-Hermitian
operators of observables in applied quantum theory. Keeping in mind
the presence of many obstacles and mathematical puzzles in this
field, we expect that the present clarification of at least some of
them might re-encourage the mathematically sufficiently rigorous
further search for the representations of quantum systems in which
the inner product in the standard Hilbert space in nontrivial. In
particular, our present results might encourage a return to all of
those constructions where the physical inner products proved
formally represented by unbounded metric operators $\Omega$ while
their formal factorization $\Theta=\Omega^\dagger\Omega$ into
``microscopic" Dyson maps would still be difficult to perform.

\section*{Acknowledgements}

F.B. acknowledges M.I.U.R. for financial support. M.Z. acknowledges
the support by the GA\v{C}R grant Nr. P203/11/1433, by the M\v{S}MT
``Doppler Institute" project Nr. LC06002 and by the Institutional
Research Plan AV0Z10480505.

\end{document}